%% file: qcl.tex
\documentclass{article}
\usepackage[a4paper,margin=1.25in]{geometry}

\usepackage{graphicx} %
\usepackage{mathtools}
\usepackage{xcolor}
\usepackage{amsmath,amssymb,amsfonts,gensymb,amsthm}
\usepackage{physics}
\usepackage{stmaryrd}
\usepackage{algorithm}
\usepackage{algorithmicx}
\usepackage{algpseudocode}
\usepackage{hyperref}
\usepackage[capitalize]{cleveref}
\definecolor{OliveGreen}{rgb}{0,0.6,0}
\hypersetup{
  colorlinks=true,
  urlcolor=blue,
  linkcolor=blue,
  citecolor=OliveGreen
}

\title{Quantum Catalytic Space}
\author{Harry Buhrman \\ Quantinuum London \& QuSoft \\ \texttt{h.m.buhrman@uva.nl} \and
Marten Folkertsma \thanks{Supported by the Dutch Ministry of Economic Affairs and Climate Policy (EZK), as part of the Quantum Delta NL programme.} \\ CWI \& QuSoft \\ \texttt{mjf@cwi.nl} \and
Ian Mertz \thanks{Supported by the Grant Agency of the Czech Republic under the grant agreement no. 24-10306S and by the Center for Foundations of Contemporary Computer Science
(Charles Univ. project UNCE 24/SCI/008).} \\ Charles University \\ \texttt{iwmertz@iuuk.mff.cuni.cz} \and
Florian Speelman \thanks{Supported by the Dutch Ministry of Economic Affairs and Climate Policy (EZK), as part of the Quantum Delta NL program, and the project Divide and Quantum `D\&Q' NWA.1389.20.241 of the program `NWA-ORC', which is partly funded by the Dutch Research Council (NWO)} \\ University of Amsterdam \& QuSoft \\ \texttt{f.speelman@uva.nl} \and
Sergii Strelchuk \thanks{Supported by a Royal Society University Research Fellowship and the EPSRC  (RoaRQ), Investigation 005 [grantreference EP/W032635/1].} \\ University of Oxford \\ \texttt{sergii.strelchuk@cs.ox.ac.uk} \and
Sathyawageeswar Subramanian \thanks{Supported by a Royal Commission for the Exhibition of 1851 Research Fellowship.} \\ University of Cambridge \\ \texttt{ss2310@cam.ac.uk} \and
Quinten Tupker \thanks{Supported by the Dutch National Growth Fund (NGF), as part of the Quantum Delta NL program.} \\ CWI \& QuSoft \\ \texttt{qmtupker@gmail.com}}
\date{\today}

\renewcommand{\P}{\mathsf{P}}

\newcommand{\coNL}{\mathsf{coNL}}

\newcommand{\DQCOne}{\mathsf{DQC_1}}

\newcommand{\CSPACE}{\mathsf{CSPACE}}
\newcommand{\QCSPACE}{\mathsf{QCSPACE}}
\newcommand{\BQCSPACE}{\mathsf{BQCSPACE}}
\newcommand{\QCSPACEC}{\mathsf{QCSPACEC}}
\newcommand{\BQCSPACEC}{\mathsf{BQCSPACEC}}
\newcommand{\QCLM}{\mathsf{QCLM}}
\newcommand{\BQCLM}{\mathsf{BQCLM}}
\newcommand{\QCSPACEM}{\mathsf{QCSPACEM}}
\newcommand{\BQCSPACEM}{\mathsf{BQCSPACEM}}
\newcommand{\PSPACE}{\mathsf{PSPACE}}
\newcommand{\mO}{\mathcal O}

\newcommand{\QL}{\mathsf{QL}}
\newcommand{\QCL}{\mathsf{QCL}}
\newcommand{\QUCL}{\mathsf{Q_UCL}}
\newcommand{\BQCL}{\mathsf{BQCL}}
\newcommand{\BQCLC}{\mathsf{BQCLC}}
\newcommand{\BQL}{\mathsf{BQL}}

\newcommand{\CL}{\mathsf{CL}}
\newcommand{\SPACE}{\mathsf{SPACE}}

\newcommand{\QNCo}{\mathsf{QNC}^1}

\newcommand{\PL}{\mathsf{PL}}

\newcommand{\EQP}{\mathsf{EQP}}
\newcommand{\BQP}{\mathsf{BQP}}
\newcommand{\QCLC}{\mathsf{QCLC}}

\newcommand{\ZPP}{\mathsf{ZPP}}
\newcommand{\poly}{\mathsf{poly}}

\newcommand{\NL}{\mathsf{NL}}
\newcommand{\BPL}{\mathsf{BPL}}
\newcommand{\TCone}{\mathsf{TC}^1}

\newcommand{\QCC}{\mathsf{QCC}}

\newtheorem{theorem}{Theorem}
\newtheorem{lemma}{Lemma}
\newtheorem{fact}{Fact}
\newtheorem{corollary}{Corollary}
\newtheorem{remark}{Remark}
\newtheorem{definition}{Definition}

\newcommand{\kb}[1]{\ket{#1}\bra{#1}}

\begin{document}

\maketitle

\begin{abstract}
\input{sections/abstract.tex}
\end{abstract}

\tableofcontents

\input{sections/introduction.tex}

\input{sections/preliminaries.tex}

\input{sections/qcl_def.tex}

\input{sections/qcl_upper_bounds.tex}

\input{sections/TC1_in_QCL_new_old_final_final.tex}

\input{sections/CL_in_DQC1.tex}

\input{sections/acknowledgements.tex}

\bibliographystyle{alpha}
\bibliography{bibliography}

\end{document}

%% file: sections/abstract.tex
Space complexity is a key field of study in theoretical computer science. In the quantum setting there are clear motivations to understand the power of space-restricted computation, as qubits are an especially precious and limited resource.

Recently, a new branch of space-bounded complexity called catalytic computing has shown 
that reusing space is a very powerful computational resource, especially for subroutines that incur little to no space overhead. While quantum catalysis in an information theoretic context, and the power of ``dirty'' qubits for quantum computation, has been studied over the years, these models are generally not suitable for use in quantum space-bounded algorithms, as they either rely on specific catalytic states or destroy
the memory being borrowed. 

We define the notion of catalytic computing in the quantum setting and show a number of initial results about the model. First, we show that quantum catalytic logspace can always be computed quantumly in polynomial time; the classical analogue of this is the largest open question in catalytic computing. This also allows quantum catalytic space to be defined in an equivalent way with respect to circuits instead of Turing machines. We also prove that quantum catalytic logspace can simulate log-depth threshold circuits, a class which is known to contain (and believed to strictly contain) quantum logspace, thus showcasing the power of quantum catalytic space. Finally we show that both unitary quantum catalytic logspace and classical catalytic logspace can be simulated in the one-clean qubit model.

%% file: sections/introduction.tex
\section{Introduction}

Space is one of the cornerstones of theoretical computer science, and the study of space-bounded computations has been crucial in the development of complexity theory. Investigating logspace computations revealed the limits of efficient computation under memory constraints and has led to striking results such as Savitch’s theorem~\cite{savitch1970relationships} and $\NL = \coNL$~\cite{immerman1988nondeterministic,szelepcsenyi1988method}. Logspace reductions are essential in classifying problems as $\NL$-complete or $\P$-complete, and leading to techniques for efficient parallelization and algorithm design.

Many graph and database problems rely on logspace techniques, making them relevant for query optimization, data retrieval, and formal verification. Furthermore, logspace computations have practical applications in streaming algorithms, embedded systems, cryptography, and model checking, where minimizing memory usage is critical.  

The emergence of quantum computing has led to remarkable theoretical speedups over the best known classical algorithms. The promise of exponential computational advantage in using principles of quantum mechanics to process information comes with formidable experimental challenges of building and maintaining quantum computers that can implement long sequences of coherent operations. This led to a renewed interest in the structure of quantum space.

\subsection{Space in quantum computation}

Understanding the true extent of the power of quantum computing in a variety of space-constrained settings is a major challenge. In contrast to the classical setting where adding a reasonable amount of extra memory to support computations is routinely achievable, producing and maintaining multiple qubits is exceptionally difficult due to several fundamental physical, engineering, and scalability issues.
Qubits are fragile and susceptible to decoherence, and maintaining long coherence times becomes significantly harder as the number of qubits increases. Furthermore, quantum error rates scale with the number of qubits, making fault-tolerant quantum computing a major challenge. In the quantum computational setting, space thus comes at a premium, and increasing the amount of space available for computation requires overcoming fundamental challenges to reduce error rates, increase control precision, and maintain entanglement across multiple systems, to name but a few.

The characterization of quantum logspace ($\QL$) and the study of the computational power of bounded-error quantum logarithmic space ($\mathsf{BQL}$) and its relationship to classical complexity classes was first done by Watrous~\cite{watrous1998space}, where it was established that $\mathsf{BQL} \subseteq \mathsf{P}$. This showed that any problem solvable in quantum logspace with bounded error is also solvable in polynomial time by a classical deterministic machine. In later work, Watrous~\cite{watrous2001quantum} showed that $\mathsf{QSPACE}(s) \subseteq \SPACE[O(s^2)]$ for all $s \geq \log n$, even when the quantum machine is allowed to err with probability arbitrarily close to $1/2$; this confirms that quantum logspace computations remain simulable within polynomial space, and is consistent with classical space complexity results such as Savitch’s theorem. His work also established that quantum logspace can efficiently solve certain algebraic problems, including the \textit{group word problem for solvable groups}, which lacks efficient classical logspace algorithms \cite{watrous2001quantum}.

These above obstacles prompted the search for extra ingredients which could lift restricted models of quantum computation (for example -- realized by quantum circuits which are classically efficiently simulatable) to regain the power of universal quantum computation. These extra ingredients (e.g.\ magic state injection) are usually studied in the context of unrestricted space and there has as of yet been no attempt to investigate them under space restrictions.   

On the other hand, there have been several notable results that illuminate various properties of quantum logspace. One of the earliest findings shows that any quantum computation that can be performed with logarithmic space can also be efficiently simulated using matchgate circuits of polynomial width, and vice versa~\cite{jozsa2010matchgate}. Following this characterisation, there have been a series of further results indicating that quantum logspace describes a non-trivial class of computations. Ta-Shma~\cite{ta2013inverting} showed that given a matrix with a bounded condition number, a quantum logspace algorithm can efficiently approximate its inverse or solve linear systems. Girish, Raz, and Zhan~\cite{girish2020quantum} described a quantum logspace algorithm to compute powers of an matrix with bounded norm and prove that deterministic logspace is equal to reversible logspace. Recently, it was shown by the same authors that the class of decision problems solvable by a quantum computer in logspace admits an efficient verification procedure~\cite{girish2024quantum}; moreover, they also show that every language in $\BQL$ has an (information-theoretically secure) streaming proof with a quantum logspace prover and a classical logspace verifier. This hints at a curious interplay between the powers of classical and quantum logspace.

\subsection{Catalysis and space}

Catalysis is a concept well-studied in the context of quantum information and is widely recognized for its counterintuitive abilities to enable (state) transformations that are otherwise infeasible (see survey by Lipka et al.~\cite{lipka2024catalysis}). A related concept, known as catalytic embedding, was recently introduced in the context of circuit synthesis as an alternative to traditional gate approximation methods in quantum circuit design~\cite{amy2023catalytic}. Here the goal is to implement a desired unitary operation {\it more efficiently} (e.g., with fewer gates, lower depth, or using a restricted gate set) than would be possible without assistance. It involves a specific, known, and often small catalyst state that is chosen to aid a particular unitary implementation.

These foregoing lines of work focus on the idea that a specific unitary may be implemented more efficiently if a special state (i.e.\ catalyst) is available, often discussing resource theories, and do not dwell on complexity theoretic implications. 

In this work, we initiate the complexity-theoretic study of the effect of catalytic space in quantum computations. Much like magic state injection is able to promote and increase quantum computational power in the space-unrestricted setting, the presence of a catalyst in the form of an extra register of quantum memory---albeit memory that already contains some stored quantum information---holds a similar promise for space-bounded quantum computations. The notion of catalytic space can be regarded as a theoretical model of qubit reusal.

The first step towards a rigorous study of catalytic logspace quantum computations is to formalize the model and means of interaction with the catalytic space. Identifying new computational capabilities endowed by the presence of a catalyst in the form of additional quantum memory, which however contains an arbitrary unknown quantum state, appears to be a significantly more challenging task due to the nature of quantum information and the inherent limitations of quantum resources. For example, any framework for quantum catalytic space must incorporate the possibility of entanglement and its inherent limitations (e.g.\ monogamy) between the catalytic memory and the rest of the work space. It has to further account for the irreversibile nature of quantum measurement. 

Remarkably, it was recently shown that the addition of a similar notion of catalytic space has major implications even  in the classical logspace setting. Buhrman et al.~\cite{buhrman2014computing} introduced a model of space, called \textit{catalytic computing}, which studies the power of ``imperfect'' memory. In addition to the usual Turing machine work tape, a catalytic machine is equipped with a much larger \textit{catalytic} work tape, which is filled with an arbitrary initial string $\tau$ and which must be reset to the configuration $\tau$ at the end of its computation.

The setting of most interest to us is \textit{catalytic logspace} ($\CL$), wherein a logspace machine is given access to a polynomial size catalytic tape. On the positive side, \cite{buhrman2014computing} showed that such machines have significantly greater power than traditional logspace, capturing the additional power of both non-determinism ($\NL$) and randomness ($\BPL$); in fact, they showed that $\CL$ can simulate the much larger class of \textit{logarithmic-depth threshold circuits} ($\TCone$). On the negative side, they also showed that $\CL$ can be simulated by \textit{(zero-error) randomized polynomial-time machines} ($\ZPP$), which are strongly believed to be much weaker than e.g.\ polynomial space.

Since then, many works have studied classical catalytic space from a variety of angles, including further results on the power of $\CL$~\cite{CookLiMertzPyne25,AgarwalaMertz25,AlekseevFilmusMertzSmalVinciguerra25} augmenting catalytic machines with other resources such as randomness or non-determinism~\cite{BuhrmanKouckyLoffSpeelman18,DattaGuptaJainSharmaTewari20,CookLiMertzPyne25,KouckyMertzPyneSami25}, considering non-uniform models such as catalytic branching programs or catalytic communication complexity~\cite{Potechin17,CookMertz22,PyneSheffieldWang25}, analyzing the robustness of classical catalytic machines to alternate conditions~\cite{BisoyiDineshSarma22,BisoyiDineshRaiSarma24,GuptaJainSharmaTewari24}, and so on. Many properties of catalytic computation have emerged that appear ripe for use in the quantum setting, such as \textit{reversibility}~\cite{Dulek15,CookLiMertzPyne25}, \textit{robustness}~\cite{GuptaJainSharmaTewari24,FolkertsmaMertzSpeelmanTupker25}, and \textit{average-case runtime bounds}~\cite{buhrman2014computing}.

Perhaps most important to motivate our current study, the utility of classical catalytic computation has been strikingly demonstrated in its use as a subroutine in an ordinary space-bounded computation: avoiding linear blowups in space when solving many instances of a problem. The most impactful result is the Tree Evaluation algorithm of Cook and Mertz~\cite{Cook24}, which was the key piece in Williams' recent breakthrough on time and space~\cite{Williams25}. Catalytic subroutines of this kind are even more relevant in the quantum setting, as they may lead to a persistent reduction of the qubit count when executing a quantum algorithm.

\subsection{Summary of results}

In this paper we initiate the systematic study of catalytic
techniques in the quantum setting. To this end
we codify a concrete definition of quantum catalytic space ($\QCSPACE$),
explore the degrees to which the definition is robust,
and establish the relationship of quantum catalytic logspace ($\QCL$)
to various classical and quantum complexity classes.

Our main technical contribution is to show that, somewhat surprisingly,
quantum Turing machines and quantum circuits
are equivalent even in the catalytic space setting:
\begin{theorem} \label{thm:cqtm_cqckt_easy}
    Let $L$ be a language, and let $s := s(n)$ and $c := c(n)$.
    Then $L$ is computable by a quantum catalytic Turing machine
    with work space $O(s)$ and catalytic space $O(c)$
    iff $L$ is computable by a family of quantum catalytic circuits
    with work space $O(s)$ and catalytic space $O(c)$.
\end{theorem}
While this translation is straightforward in other settings, $\QCL$
has no \textit{a priori} polynomial time bound, and so there is no obvious
way to define the length of a catalyic circuit without running into trouble.
However, we prove that the result of Buhrman et al.~\cite{buhrman2014computing}
which shows that $\CL$ takes polynomial time on average can be strengthened
in the quantum case, to show that $\QCL$ \textit{always} takes polynomial time
without any error:
\begin{theorem}
\label{thm:QCL_in_EQP}
  $\QCL \subseteq \mathsf{EQP}$
\end{theorem}
\noindent
We find Theorem~\ref{thm:QCL_in_EQP} intriguing for many reasons.
Naturally it is exciting to be able to solve the ``holy grail'' of
catalytic computing in the quantum setting.
The story of classical catalytic computing has been the ability of
clever algorithms to circumvent the resetting condition of the catalytic tape
and use it for powerful purposes, but Theorem~\ref{thm:QCL_in_EQP} shows that
conversely, the additional power of quantum techniques in such algorithms
does not offset the additional restrictiveness of resetting a quantum state.
Quantum computation is a model fundamentally built on reversible instructions,
with the one exception being the final measurement with which we obtain our answer;
Theorem~\ref{thm:QCL_in_EQP} shows that this measurement is a massive obstruction
to reversibility, as having access to such a huge resource with only the reversible
restriction---something which is taken care of in the intermediate computation
already---gives less power than we initially assumed.

In terms of class containments, we focus on two questions: the relationship
of quantum and classical catalytic space, and the relationship of catalytic space to
the one-clean qubit model ($\DQCOne$), a pre-existing object of study in
quantum complexity which bears a strong resemblance to catalysis. We show
that, while $\CL \subseteq \QCL$ is surprisingly out of reach at the moment,
this can be shown for an important subclass of $\CL$, one which captures the strongest known
classical containment:
\begin{theorem} \label{thm:tcone}
    $\TCone \subseteq \QCL$
\end{theorem}
\noindent
As a consequence, we show that $\TCone$ constitutes
a natural class of functions for which catalysis gives additional power to
quantum computation.

We also show that unitary $\QCL$
($\QUCL$) and classical $\CL$ are both contained in $\DQCOne$:
\begin{theorem}
\label{thm:bqucl_in_dqc1}
$\mathsf{BQ_{U}CL} \subseteq \DQCOne$ 
\end{theorem}
\begin{theorem}
\label{thm:cl_in_dqc1}
    $\CL \subseteq \DQCOne$
\end{theorem}
\noindent
Note that we use a version of $\DQCOne$ defined using a logspace controller instead of a polynomial time controller as may also be done. These results show how much of the power of $\DQCOne$ comes from avoiding
the limitation of the resetting condition on the ``dirty'' work space.

\subsection{Open problems}

We identify a number of interesting avenues to further explore the power of
quantum catalytic space, and understand its relation to various (quantum) complexity classes.

\paragraph{$\QCL$ subroutines.} Remarkably, classical catalytic subroutines can already be used to achieve analogous space savings in $\QCL$. Is it possible to identify genuinely quantum subroutines to achieve savings beyond those attained by classical generalizations? This is not so straightforward because the subset of qubits being reused in a catalytic subroutine could become entangled with qubits that cannot be accessed by the subroutine. Therefore, there might be a non-trivial and inaccessible reference system with respect to which the catalytic property must hold. While we show the presence of such an inaccessible reference system does not change the model we define, designing quantum catalytic subroutines (cf.\ classical results in  \cite{Cook24,Williams25}) stands out as a fertile direction for future work.

\paragraph{$\CL$ vs $\QCL$.} While we have started investigating the question, we still have no simple answer as to the relationship between $\CL$ and $\QCL$; in fact we have not even ruled out that $\CL$ contains $\QCL$. The primary challenge is that while any $\CL$ machine runs in polynomial time in expectation over the catalytic tape, $\QCL$ machines \emph{always} run in polynomial time. We do not know how to fit in pathological cases where $\CL$ runs in exponential time, for example, into $\QCL$. Similarly, a problem or oracle that can separate $\QCL$ from $\CL$ would also be of interest\footnote{A candidate oracle for showing a separation between $\QCL$ and $\CL$ is the oracle relative to which $\CL$ and $\PSPACE$ are equivalent, as shown in \cite{buhrman2014computing}. This oracle uses the fact that the initial catalytic tapes of $\CL$ are either compressible or random, using the oracle differently for either situation. This type of adaptive usage of the oracle, based on the given catalytic state, seems not to translate to the quantum setting due to Theorem~\ref{thm:catalytic_independent_time}.}.

\paragraph{$\QNCo$ vs $\QCL$.} Starting with Barrington's Theorem~\cite{Barrington89}, a landmark result in space complexity, a classical line of work~\cite{BenorCleve92,buhrman2014computing} has shown that polynomial-size formulas over many different gatesets can be computed using only logarithmic space, using a reversible, algebraic characterization of computation. Such a result in the quantum case, i.e.\ $\QNCo \subseteq \QL$, appears far out of reach, as this would imply e.g.\ novel derandomizations in polynomial time. However, such techniques are also key to the study of catalytic computation, and so
perhaps we can show $\QNCo$ or a similar quantum circuit class is contained in $\QCL$. This would give a clear indication of the power of quantumness in catalytic computation.

\paragraph{$\QCL$ vs $\DQCOne$.} While we seem to find that $\QUCL$ or $\QCL$ without intermediate measurements is contained in $\DQCOne$, it is unclear if this still holds when we allow intermediate measurements.

\paragraph {$\QCL$} with errors. One aspect of our results which is discordant with the usual mode of quantum computation is that we require the catalytic tape be \textit{exactly} reset by the computation. On the other hand, many basic primitives in quantum computing, such as converting between gatesets, can introduce errors into the computation, and in practice even the ambient environment can be assumed to cause such issues. Thus it seems natural to study the power of $\QCL$ when we allow a small, potentially exponentially small, trace distance between the initial and final catalytic states. This model is well-understood in the classical world~\cite{GuptaJainSharmaTewari24,FolkertsmaMertzSpeelmanTupker25}, but it would be interesting to see whether our techniques can be made robust to this small error or, to the contrary, whether this slight relaxation is enough to overcome the barriers in our work, chiefly the inability to show $\CL \subseteq \QCL$.

%% file: sections/preliminaries.tex
\section{Preliminaries}

\subsection{Quantum computation}

For this work we will consider complex \emph{Hilbert} spaces $\mathcal H \cong \mathbb C^d$ of dimension $d$, that will form the state space for a quantum system. Multiple quantum systems are combined by taking the tensor product of their Hilbert spaces, such as $\mathcal H_1 \otimes \mathcal H_2$. We will often write $\mathcal H_s$ to denote the Hilbert space $\left(\mathbb{C}^2\right)^{\otimes s}$ of $s$ qubits, where the dimension is given by function $d(\mathcal H_s) = 2^s$. We will also often use the abbreviation $[n] = \{1,\dots,n\}$. Below, we recall some of the important background required for this article, referring the reader to \cite{NielsenChuang2010} for more details.

\begin{definition}[Quantum states]
A \emph{pure quantum states} is a unit vector of the Hilbert space $\ket{\psi} \in \mathcal H$, with the normalization condition $\braket{\psi} = 1$. 
We also make use of more general states represented by \emph{density matrices} $\rho$ which are positive semi definite operators on a Hilbert space with unit trace, $Tr[\rho] = 1$. Density matrices describe \emph{mixed states} which, beyond pure quantum states, can also capture classical uncertainty. In other words, they correspond to classical mixtures of pure quantum states. The density matrix of a pure state is $\rho = \ket{\psi}\bra{\psi}$. Given an ensemble of states $\{\ket{\psi_i}\}$ and corresponding probabilities $\{p_i\}$, with $p_i \geq 0$ and $\sum_i p_i = 1$, it can be represented by a mixed state of the form $\rho=\sum_i p_i \ket{\psi_i}\bra{\psi_i}$. We will denote the set of mixed states a Hilbert space $\mathcal H$ by $D(\mathcal H)$.
\end{definition}

\begin{definition}[Quantum channels]
A \emph{quantum channel} is a linear operator that maps density matrices to density matrices, $\Phi: D(\mathcal H_1) \rightarrow D(\mathcal H_2)$ (also known as superoperators or CPTP maps). It is also required to have two additional properties: 1) it must be completely positive; and 2) it must be trace preserving. We denote the set of channels from $D(\mathcal H)$ to itself by $\mathcal C(D(\mathcal H))$.
\end{definition}

We denote the identity channel on $d$ qubits by $\mathcal I_d$, or just $\mathcal I$ when $d$ is clear from context. The \emph{Choi matrix} of a channel $\Phi$ that acts on an input space $\mathcal H$ of dimension $d$ is defined by the action of $\Phi$ on the first register of a maximally entangled state in $\mathcal H\otimes\mathcal H$
\begin{equation*}
    J(\Phi)
    \coloneqq \left({\Phi \otimes \mathcal I_d}\right) \left(\frac1d \sum_{i,j=1}^d \ket{i}\bra{j} \otimes \ket{i}\bra{j}\right)
    = \frac1d \sum_{i,j=1}^d \Phi \biggl(\ket{i}\bra{j}\biggr) \otimes \ket{i}\bra{j}.
\end{equation*}

\begin{definition}
The trace distance between two density matrices $\rho,\sigma \in D(\mathcal H)$ is defined by:
\[
    ||\rho - \sigma||_1 = Tr[\sqrt{(\rho -\sigma)^\dagger(\rho - \sigma)}],
\]
where $A^\dagger$ denotes the conjugate transpose of the matrix $A^\dagger = \bar{A}^T$.  
\end{definition}

\noindent
It is well known that no physical process can increase the trace distance between two states:
\begin{lemma}[Contractivity under CPTP maps {\cite[Theorem 9.2]{NielsenChuang2010}}]
\label{lem:non-increasing trace}
Let $\Phi \in \mathcal C(D(\mathcal H))$ and $\rho, \sigma \in D(\mathcal H)$ then the trace distance between $\rho$ and $\sigma$ can not increase under application of $\Phi$:
\[
    ||\Phi(\rho) - \Phi(\sigma)||_1 \leq ||\rho - \sigma||_1 
\]
\end{lemma}

\subsubsection{Quantum Turing machines}

Our fundamental computation model in quantum computing will be the quantum analogue
of Turing machines \cite{deutsch1985quantum,Bernstein1997}, which we define informally below.
\begin{definition}[Quantum Turing machine]
    A \emph{quantum Turing machine} is a classical Turing machine with an additional quantum tape and quantum register. The quantum register does not affect the classical part of the machine in any way, except in that the qubits in the quantum register can be measured in the computational basis. On doing so, the values read from the measurement are copied into the classical registry, from where they can be used to affect the operation of the machine. The quantum Turing machine can perform any gate from its quantum gate set on its quantum registry. We assume this gate set is fixed and universal. Finally, the tape head on the quantum tape can swap qubits between the quantum registry and the position that the quantum tape head is located at. This applies a two-qubit $\mathrm{SWAP}$ gate.
\end{definition}

We define a number of complexity classes with respect to efficient computation
by quantum Turing machines \cite{Bernstein1997,Nishino2002}\footnote{We do not attempt to provide an exhaustive list of references to the vast literature on this topic, and refer the interested reader to the \href{https://complexityzoo.net/Complexity_Zoo}{Complexity Zoo} for such a list.}.

\begin{definition}[$\BQP$] $\BQP$ is the set of all languages $L = (L_{\text{yes}},L_{\text{no}})  \subset \{ 0,1\}^* \times \{0, 1\}^*$ for which there exists a quantum Turing machine $M$ using $t=\poly(n)$ time such that for every input $x\in L$ of length $n=|x|$,
\begin{itemize}
    \item if $x \in L_{\text{yes}} $ then the probability that $M$ accepts input $x$ is $\geq c$,
    \item if $x \in L_{\text{no}}$ then the probability that $M$ accepts input $x$ is $\leq s$.
\end{itemize}
\end{definition}

\begin{definition}[$\BQL$] $\BQL$ is the set of all languages $L = (L_{\text{yes}},L_{\text{no}})  \subset \{ 0,1\}^* \times \{0, 1\}^*$ for which there exists a quantum Turing machine $M$ using $r=O(\log(n))$ quantum and classical space such that for every input $x\in L$ of length $n=|x|$,
\begin{itemize}
    \item if $x \in L_{\text{yes}} $ then the probability that $M$ accepts input $x$ is $\geq c$,
    \item if $x \in L_{\text{no}}$ then the probability that $M$ accepts input $x$ is $\leq s$.
\end{itemize}
\end{definition}
The completeness and soundness parameters in both the above definitions can be chosen to be $c=2/3$ and $s=1/3$ without affecting the set of languages.
\begin{definition}[$\EQP$] $\EQP$ is the set of all languages $L = (L_{\text{yes}},L_{\text{no}})  \subset \{ 0,1\}^* \times \{0, 1\}^*$ for which there exists a quantum Turing machine $M$ using $t=\poly(n)$ time such that for every input $x\in L$ of length $n=|x|$,
\begin{itemize}
    \item if $x \in L_{\text{yes}} $ then $M$ outputs one with certainty on measurement,
    \item if $x \in L_{\text{no}}$ then $M$ output zero with certainty on measurement.
\end{itemize}
\end{definition}

\begin{remark}
Note that the definition of $\EQP$ is gateset dependent; this is due to the fact that quantum gatesets only allow universality up to approximation, which means that if a quantum complexity class requires perfect soundness and completeness, as does $\EQP$, it also has to be gateset dependent. 
\end{remark}

\subsubsection{Quantum circuits}

We may also define quantum complexity classes using uniform quantum circuits.
For this we use similar definitions to those provided by \cite{fefferman_remscrim},
which readers may refer to for more details.

\begin{definition}
\label{def:q-ckts}
    Let $s := s(n), t := t(n), k := k(n)$, let $\mathcal{K}$ be a family of machines,
    and let $\mathcal{G}$ be a set of $k$-local operators.
    A $\mathcal{K}$-uniform space-$s$ time-$t$ family of quantum circuits
    over $\mathcal{G}$ is a set $\{Q_x\}_{x \in \{0,1\}^n}$, where
    each $Q_x$ is a sequence of tuples
    $\langle i,g,j_1 \ldots j_k \rangle \in [t] \times \mathcal{G} \times [s]^k$
    such that there is a deterministic TM $M \in \mathcal{K}$ which,
    on input $x \in \mathcal{X}$, outputs a description of $Q_x$.
    
    The execution of $Q_x$ consists of initializing a vector $\ket{\psi}$ to
    $\ket{0^s}$ within $\mathcal{H}_s$ and applying, for each step $i \in [t]$ in order,
    each gate $g$ to qubits $j_1 \ldots j_k$ such that
    $\langle i,g,j_1 \ldots j_k \rangle \in Q_x$. The output of $Q_x$ is the
    value obtained by measuring the first qubit at the end of the computation.
    
    If $\mathcal{G}$ consists of unitary operators, we call these \textit{unitary circuits}
    and call each $g$ a \textit{gate}.
    If $\mathcal{G}$ additionally consists of measurements together with postprocessing
    and feed forward by (classical) $\mathcal{K}$-machines, we call these
    \textit{general circuits} and call each $g$ a \textit{channel}.
\end{definition}

It is known that polynomial-time uniform general quantum circuits over $n$ qubits with $\poly(n)$ gates can be used to provide an alternative definition of $\BQP$ \cite{Yao93}. Similarly, logspace uniform general quantum circuits of logarithmic width can be used as an alternative to define classes such as $\BQL$ \cite{fefferman_remscrim}.

\subsection{Catalytic computation}

We finally recall the known classical definitions of catalytic classical
computation.

\begin{definition}[\cite{buhrman2014computing}]
\label{def:cat}
    A \textit{catalytic Turing Machine} with space $s := s(n)$
    and \textit{catalytic space} $c := c(n)$ is a Turing Machine $M$ with
    a work tape of length $s$ and a \textit{catalytic} tape of length $c$.
    We require that for any $\tau \in \{0,1\}^c$, if we initialize the catalytic
    tape to $\tau$, then on any given input $x$, the execution of $M$
    on $x$ halts with $\tau$ on the catalytic tape.
\end{definition}

This definition gives rise to a natural complexity class $\CSPACE[s,c]$,
which is a variant of the ordinary class $\SPACE[s]$.
The most well-studied variant is \textit{catalytic logspace},
where $s$ is logarithmic and $c$ is polynomial.

\begin{definition}
\label{def:cspace}
    We define $\CSPACE[s,c]$ to be the class of all functions $f$ for which there exists a catalytic Turing Machine $M$ with space $s$ and catalytic space $c$ such that on input $x$, $M(x) = f(x)$.
    We further define \textit{catalytic logspace} as
    $$\CL := \bigcup_{k \in \mathbb{N}} \CSPACE(k \log n, n^k)$$
\end{definition}

%% file: sections/qcl_def.tex
\section{Quantum catalytic space}
\label{sec:qcm_def}
The first goal of this paper is to find a proper definition of quantum catalytic space.
There are many choices that have to be made in the model, but we begin with our general
definition up front, leaving questions of machine model, uniformity, gateset, and initial
catalytic tapes. These will be discussed and clarified in the rest of this section.

\begin{definition}[Quantum catalytic machine] \label{def:qcm}
	A \textit{quantum catalytic machine} with work space $s := s(n)$, catalytic space
	$c := c(n)$, uniformity $\mathcal{K}$, gateset $\mathcal{G}$, and catalytic set
	$\mathcal{A}$ is a $\mathcal{K}$-uniform quantum machine $M$ with operations from
	$\mathcal{G}$ acting on two Hilbert spaces, $\mathcal H_{s}$ and $\mathcal{H}_c$,
	of dimensions $2^s$ and $2^c$ respectively.
	The latter space, called the \textit{catalytic tape},
	will be initialized to some $\rho \in \mathcal{A} \subseteq D(\mathcal{H}_c)$.
	We require that for any $\rho \in \mathcal{A}$,
	if we initialize the catalytic tape to state $\rho$,
	then on any given input $x\in\{0,1\}^n$,
	the execution of $M(x)$ halts with $\rho$ on the catalytic tape. Furthermore, we require that the output state on the worktape is independent of the catalytic state $\rho$.\footnote{We justify this final requirement in Lemma~\ref{lem:approx_out}.} The final action of the machine can be represented by a quantum channel $\Phi_{x}:\kb{0}\otimes \rho\mapsto \eta_{x} \otimes \rho$, for any catalytic state $\rho$ and input $x\in \{0,1\}^n$, and some output state $\eta$.
\end{definition}

\noindent
This gives rise to the following complexity classes:

\begin{definition}[Quantum catalytic complexity] \label{def:qcspace}
    $\QCSPACE[s,c]$ is the class of Boolean functions which can be decided
    with probability 1 by a quantum catalytic machine with work memory $s$ and
    catalytic memory $c$.
    
    $\BQCSPACE[s,c]$ is the class of Boolean functions which can be decided
    with probability 2/3 by a quantum catalytic machine with work memory $s$ and
    catalytic memory $c$.
\end{definition}

\noindent
We further specify to the case of quantum catalytic logspace:

\begin{definition}[Quantum catalytic logspace] \label{def:qcl}
    $$\QCL = \bigcup_{k \in \mathbb{N}} \QCSPACE[k \log n,n^{k}]$$
    $$\BQCL = \bigcup_{k \in \mathbb{N}} \BQCSPACE[k \log n,n^{k}]$$
\end{definition}

\subsection{Machine model}

We begin by defining the two natural choices of base model for
quantum catalytic machines, namely \textit{Turing machines} and
\textit{circuits}.

\begin{definition}[Quantum catalytic Turing machine] \label{def:qcm_tm}
	A \textit{quantum catalytic Turing machine} is defined as
	in Definition~\ref{def:qcm} with quantum Turing machines
	as our machine model. We write $\QCSPACEM$ (respectively
	$\BQCSPACEM$, $\QCLM$, and $\BQCLM$) to refer to $\QCSPACE$
	with quantum Turing machines.
\end{definition}

\begin{definition}[Quantum catalytic circuits] \label{def:qcm_ckt}
	A \textit{quantum catalytic circuit} is defined as
	in Definition~\ref{def:qcm} with time-$2^{O(s)}$ quantum
	circuits as our machine model. We write $\QCSPACEC$ (respectively
	$\BQCSPACEC$, $\QCLC$, and $\BQCLC$) to refer to $\QCSPACE$
	with quantum catalytic circuits.
\end{definition}

Given that $\CL$ and related classes are defined in terms of (classical) Turing
machines, the option of circuits seems surprising and perhaps unnatural. For example,
Definition~\ref{def:qcm_ckt} imposes a time bound as part of its definition,
while for $\CL$ there is no known containment in polynomial time.
For quantum circuits and Turing machines without access to the catalytic tape,
a simple equivalence has been known for a long time
\cite{Yao93};
however, Definition~\ref{def:qcm_ckt} only allows for circuits of length
$2^{O(s)}$, while a generic transformation on $s+c$ qubit registers would
give a circuit of length $2^{O(s+c)}$, i.e. requiring an exponential overhead. 

The main result of this paper is to show that these
models are in fact equivalent:

\begin{theorem} \label{thm:cqtm_cqckt}
    For $s = \Omega(\log n), c = 2^{O(s)}$
    $$\QCSPACEM[O(s),O(c)] = \QCSPACEC[O(s),O(c)]$$
	$$\BQCSPACEM[O(s),O(c)] = \BQCSPACEC[O(s),O(c)]$$
\end{theorem}

\noindent
For the rest of this section we will deal with all auxiliary issues, namely
the choice of catalytic tapes and gateset, for quantum circuits alone; while
all proofs can be made to hold for quantum Turing machines without much issue,
this is also obviated by Theorem~\ref{thm:cqtm_cqckt}, which we will prove in
Section~\ref{sec:tm-ckt}.

\subsection{Catalytic tapes}

We now move to discussing the choice of initial catalytic tapes $\mathcal{A}$.
Perhaps the most immediate choice would be to put no restrictions on $\mathcal{A}$
and allow our catalytic tapes to come from the set of all density matrices in
$D(\mathcal{H}_c)$; this will ultimately be our definition.

\begin{definition}
    We fix the catalytic set in Definition~\ref{def:qcm} to be
    $\mathcal{A} = D(\mathcal{H}_c)$.
\end{definition}

While this is a natural option, encompassing every possible state on $c$ qubits,
there are other choices one can make.
We propose four natural options---density matrices and three others---and
show that all four are equivalent, thus justifying our choice.

\begin{definition}
    We define the following catalytic sets:
    \begin{itemize}
    \item $\mathsf{Density}$ is the set of all density matrices $\rho \in D(\mathcal{H}_c)$.
    \item $\mathsf{Pure}$ is the set of all pure states $\ket{\psi} \in \mathcal H_c$.
    \item $\mathsf{PauliProd} = \{\ket{\mathrm{PP}} : \ket{\mathrm{PP}} = \displaystyle\bigotimes_{i=1}^c\ket{\phi}_i\}$ is the set of tensor products of eigenstates of the single-qubit Pauli operators, where $\ket{\phi}_i \in \{\ket{0},\ket{1}, \ket{+}, \ket{-},\allowbreak \ket{\circlearrowright} , \ket{\circlearrowleft}\}\subset \mathcal H_2$. 
    \item $\mathsf{EPR} = \{\frac{1}{\sqrt{2^{c}}}\sum_{i=0}^{2^c-1}\ket{i}\ket{i}\} \subset \mathcal H_c \otimes \mathcal H_c$ is the unique state of $c$ EPR pairs, where the catalytic tape will be formed of one half of each EPR pair; the other halves are retained as a reference system which cannot be operated on by the quantum circuit. the quantum circuit is of the form $Q_x = \tilde{Q}_x\otimes \mathcal I_c$, acting as the Identity on the second set of halves of the EPR pairs that is inaccessible to the circuit.
    \end{itemize}
\end{definition}

\begin{remark}
We briefly comment on the fourth set, i.e. $\mathsf{EPR}$. Using classical catalytic techniques as a subroutine has proven to be very useful, for instance in giving an algorithm for tree evaluation in $\mO(\log n \log(\log n))$ space \cite{Cook24}. One can also consider using analogous quantum catalytic techniques as subroutines for quantum computations, albeit this does not appear straightforward due to inherent quantum limitations. We will see that this complication can be effectively modeled by considering the initial state of the catalytic tape to be the halves of $c$ EPR pairs.
\end{remark}

We will now prove that the four classes of quantum catalytic circuits with initial catalytic states restricted to one of the four sets $D(\mathcal H_c)$, $\mathcal H_c$, $\mathsf{PauliProd}$, and $\mathsf{EPR}$ respectively, are all equivalent. For this we first require the following lemma.

\begin{fact}
\label{lem:RB_basis}
Any $2^d\times 2^d$ complex matrix can be written as a linear combination of rank-1 outer products of states from $\mathsf{PauliProd}$ over $d$ qubits. In other words, the complex span of the set of $d$-qubit tensor products of Pauli eigenstates equals the set of $2^d\times 2^d$ complex matrices.
\end{fact}
\begin{proof}
Note that all four Pauli matrices can be written as a linear combination of two of the Pauli eigenstates:
\begin{align*}
&I = \kb{0} + \kb{1},\quad X = \kb{+} - \kb{-}, \\
&Z = \kb{0} - \kb{1},\quad Y = \kb{\circlearrowright} - \kb{\circlearrowleft}.
\end{align*}
The four Pauli matrices form a basis for $2\times 2$ complex matrices. Consequently, Pauli strings of length $d$---i.e.,\ tensor products of $d$ Pauli matrices---form a basis for $2^d\times 2^d$ matrices.
\end{proof}

Now we can state the theorem:

\begin{theorem}
Let $\QCC_A$ denote quantum catalytic circuits with initial catalytic tapes
coming from $A$. Then
The following four classes of quantum catalytic circuits are equivalent:
\[
    \QCC_{\mathsf{Density}} = \QCC_{\mathsf{Pure}} = \QCC_{\mathsf{PauliProd}} = \QCC_{\mathsf{EPR}}
\]
\end{theorem}

\begin{proof}
First note the obvious implications: for any quantum catalytic circuit $\Phi$,
\begin{align*}
\Phi \in \QCC_{\mathsf{Density}} &\implies \Phi \in \QCC_{\mathsf{Pure}} \\
\Phi \in \QCC_{\mathsf{Pure}} &\implies \Phi \in \QCC_{\mathsf{PauliProd}}
\end{align*}
these follow due to the fact that $\mathsf{PauliProd} \subset \mathsf{Pure} \subset \mathsf{Density}$.
To finish the proof, we will further show the following two implications.
\begin{align*}
&(1)\quad \Phi \in \QCC_{\mathsf{PauliProd}} \implies \Phi \otimes \mathcal{I}_c \in \QCC_{\mathsf{EPR}}\\
&(2)\quad \Phi \otimes \mathcal{I}_c \in \QCC_{\mathsf{EPR}} \implies \Phi \in \QCC_{\mathsf{Density}}
\end{align*}
We first prove implication (1). Let $\Phi$ be a circuit from $\QCC_{\mathsf{PauliProd}}$ and consider the action of $\Phi \otimes \mathcal I_{c}$ (where the Identity operator acts on the inaccessible halves of the EPR pairs) on the state $\frac{1}{2^c}\kb{0} \sum_{i,j} \ket{i}\bra{j} \otimes \ket{i}\bra{j}$:
\[
    \Phi \otimes \mathcal I_{c}\left(\frac{1}{2^c}\kb{0} \sum_{i,j} \ket{i}\bra{j} \otimes \ket{i}\bra{j}\right) = \frac{1}{2^c} \sum_{i,j} \Phi\biggl(\kb{0}\otimes \ket{i}\bra{j}\biggr)\otimes \ket{i}\bra{j},
\]
because $\Phi$ being a channel is a linear operator. By \cref{lem:RB_basis}, $\ket{i}\bra{j}$ can be written as a linear combination of rank-1 projectors onto $\mathsf{PauliProd}$ states. Since $\Phi$ is catalytic with respect to $\mathsf{PauliProd}$, it follows that 
\[
\frac{1}{2^c} \sum_{i,j} \Phi\biggl(\kb{0}\otimes \ket{i}\bra{j}\biggr)\otimes \ket{i}\bra{j} =  \eta \otimes \frac{1}{2^c} \sum_{i,j}\ket{i}\bra{j} \otimes \ket{i}\bra{j},
\]
for some state in $\eta \in D(\mathcal H_s)$. This shows that $\Phi \in \QCC_{\mathsf{EPR}}$.

Implication (2) requires a similar approach. Let $\tilde{\Phi} \in \QCC_{\mathsf{EPR}}$, then we can write $\tilde{\Phi} = \Phi \otimes \mathcal I_c$. For a given input state $\kb{0}\in \mathcal H_s$ the action of $\Phi \otimes \mathcal I_c$ must satisfy
\[
    \Phi\left(\frac{1}{2^c} \sum_{i,j}\ket{0}\bra{0}\otimes \ket{i}\bra{j}\right) \otimes \ket{i}\bra{j} = \eta \otimes \frac{1}{2^c} \sum_{i,j} \ket{i}\bra{j} \otimes \ket{i}\bra{j},
\]
for some state in $\eta \in D(\mathcal H_s)$. 
 Since the catalytic state of $c$ EPR pairs is returned perfectly unaffected for every choice of input state, the effective channel of $\Phi$ can also be written as a tensor product channel: $\Phi = \Gamma_s \otimes \Xi_c$\footnote{It seems that the catalyst does not offer any improvement, because we can write $\Phi$ as a tensor product of the action on the logspace clean qubits and the action of the catalyst, however this does not need to hold. Only the action as a whole is writable as a tensor product, it might actually consist of intermediate steps that are not of tensor product form, therefor $\Gamma_s$ might only have an efficient circuit description in the presence of a catalyst.}, with the action of $\Xi_c$ being
\[
\frac{1}{2^c}\sum_{i,j} \Xi_c\biggl(\ket{i}\bra{j}\biggr) \otimes \ket{i}\bra{j} = \frac{1}{2^c} \sum_{i,j} \ket{i}\bra{j} \otimes \ket{i}\bra{j}\;. 
\]
Note that although the effective channel factorises into a tensor product across the work and catalytic registers, without the catalytic tape much larger circuits may be required to implement $\Gamma_c$. Moving forward, this implies that the Choi matrix of $\Xi_c$ is 
\[
    J(\Xi_c) = \sum_{i,j} \Xi_c\biggl(\ket{i}\bra{j}\biggr) \otimes \ket{i}\bra{j} = \sum_{i,j} \ket{i}\bra{j} \otimes \ket{i}\bra{j} = J(\mathcal I),
\]
and therefore the effective channel $\Xi_c$ is the identity channel. This gives that for any state $\rho \in \mathcal H_c$ it must hold that on input $\ket{0}\bra{0}$, the channel $\Phi$ must act as follows:
\[
\Phi(\ket{0}\bra{0} \otimes \rho) = \eta \otimes \rho \qedhere
\]
\end{proof}
\begin{remark}
In the proof that these channel definitions are equivalent we actually showed that any channel under one definition also furnishes an instance of the other definitions. This means that they are also operationally equivalent. These equivalence proofs therefore have to hold for any type of machine model that has to adhere to the same restrictions of resetting the input state in the catalytic space. In particular it also holds for quantum Turing machines.
\end{remark}

\subsection{Gateset}
When discussing quantum circuits, a fundamental issue is the underlying gate set. Unlike the classical case, unitary operations form a continuous space, and finite-sized circuits over finite gate sets cannot implement arbitrary unitaries. However, there do exist finite gate sets of constant locality (that is, fan-in) which are quantum universal, in the sense that any $n$-qubit unitary may be approximated to any desired precision $\epsilon$ in $\ell_2$-distance by a product of $l=O(\poly\log\frac{1}{\epsilon})$ gates from the universal gate set; this is the celebrated Solovay-Kitaev theorem \cite{Kitaev1997,Dawson06SolovayKitaev,NielsenChuang2010}. From the standpoint of complexity classes, Nishimura and Ozawa~\cite{Nishimura_2009} also showed that polynomial-time quantum Turing machines are exactly equivalent to finitely generated uniform quantum circuits.

We note that \cref{def:qcm_ckt,def:q-ckts} do not make reference any fixed universal gate set. A potential issue that arises in this regard is that the complexity class being defined may depend in an intricate way on the chosen universal gate set, since it may not be possible to perfectly reset every initial catalytic state under our uniformity and resource constraints. If we relax the notion of catalyticity to mean that the initial catalytic state only has to be reset to within $\epsilon$ trace distance at the end of the computation, one can use the Solovay-Kitaev theorem to see that every choice of gate set leads to the same complexity class in \cref{def:qcm_ckt}. This interesting model resembles classical catalytic space classes with small errors in resetting, and we leave it as an open question to determine how it relates to the exact resetting model.

Returning to our setting that requires the quantum catalytic machine to perfectly reset the catalytic space to its initial state at the end of the computation, we will restrict out attention to the case of universal quantum gate sets that are infinite (for the complexity theoretic properties of circuit families over such gate sets, see e.g. \cite{Nishimura02infinitegates}). In this case, our definition is robust to the choice of gate set since any unitary may be implemented exactly by finite-sized circuits over such gate sets. Consequently changing the gate set does not change the set of catalytic states that can be reset exactly by the machine. This results in well-defined catalytic complexity classes independent of the specific choice of gate set.

\subsection{Uniformity}

Similar to gatesets,
the question of uniformity is quite subjective,
as different uniformity conditions will lead to different levels
of expressiveness for our machines.

\begin{definition}
    We fix the uniformity in Definition~\ref{def:qcm} to be
    $\mathcal{K} = \SPACE[O(s)]$.
\end{definition}

We choose $\SPACE[O(s)]$ as it is the largest class of classical machines
a $\QCSPACE[s,c]$ machine should seemingly contain by default.
Thus we believe the choice of $\SPACE[O(s)]$-uniformity is best
suited to removing classical uniformity considerations from taking the
forefront of the discussion regarding quantum catalytic space.

The question of how uniformity affects the power of $\QCSPACE$
is left to future work; we only comment briefly here on natural alternative choices.
Perhaps the most immediate would be to
consider $\CSPACE[s,c]$ uniformity, as it mirrors our quantum machine.
As we will see later, it is not clear how to prove $\QCSPACE[s,c]$ contains
$\CSPACE[s,c]$ directly, an interesting technical challenge that would be
rendered moot by building it into the uniformity. Similarly we avoid
$\P$-uniformity because it is not known, and even strongly
disbelieved, that $\CL$ contains $\P$.

%% file: sections/qcl_upper_bounds.tex
\section{$\QCL$ upper bounds}
\label{sec:tm-ckt}

In this section we will finally return to the question of our quantum
machine model, showing that Turing machines and circuits are equivalent.
One major stepping stone is to show that quantum catalytic Turing machines adhere
to a polynomial runtime bound for \textit{all} possible
initializations of the catalytic tape.

Before all else, a remark is in order as to why such a restriction should hold for a seemingly stronger
model, i.e. $\QCLM$, when it is not in fact known for $\CL$.
While quantum catalytic space has access to more powerful computations, i.e.
quantum operations, it also has the much stronger restriction of resetting arbitrary
density matrices rather than arbitrary bit strings. This restriction gives rise to
a much stronger upper bound argument, and in fact rules out one of the main techniques
available to classical Turing machines, namely compression arguments (see c.f.
\cite{Dulek15,CookLiMertzPyne25}).

\subsection{Polynomial average runtime bound}

We begin by showing an analogue of the classical result of
\cite{buhrman2014computing},
i.e. the average runtime of a quantum catalytic machine for a random initial catalytic
state $\rho$ is polynomial in the number of work qubits.
We note that the runtime of a quantum Turing machine need not be a deterministic function of the input;
$M$ has access to quantum states and intermediate measurements,
from which it is possible to generate randomness which
might influence the time that machine takes to halt.

\begin{definition}
	Given a quantum catalytic Turing machine $M$, a fixed input $x \in \{0,1\}^n$,
	and an initial catalytic tape $\rho$, we denote by $T(M,x,\rho)$
	the distribution of runtimes of $M$ on input $x$ and initial catalytic tape $\rho$.
\end{definition}

For an averaging argument to hold, we need to have a quantum notion of non-overlapping configuration graphs. 

\begin{lemma}
\label{lem:orthogonality-average}
	Let $M$ be a quantum catalytic Turing machine, and let
	$\{\tau_i\}_i$ form an orthonormal basis for
	$D(\mathcal H_c)$. For all $i$ and $t$,
	let $\rho_{i,t}$ be the density matrix
	describing the state of the classical tape, quantum tape, and internal state
	of $M$ at time step $t$ on initial catalytic tape $\tau_i$.
	Then if $M$ is absolutely halting, all elements of the set
	$\{\rho_{i,t}\}_{i,t}$ are orthogonal.
\end{lemma}

\begin{proof}
	We first consider the states $\rho_{i,t}$ for a fixed $i$.
	Assume instead that there exists some times $t$ and $t'$ where the states are not orthogonal.
	This means that the state at time step $t$ can be written as a superposition between
	the state in time step $t'$ and the state $\rho_{i,t} = p \rho_{i,t'} + (1 - p) \eta$
	for some $p > 0$. This forms a loop in the configuration graph where part of the
	state is back at time step $t'$. The amplitude of the part of the state in this
	loop will shrink over time, but never go to zero. The part of the state that is
	stuck in the loop will never reach the halting state, therefore this is in
	contradiction with the assumption that the quantum Turing machine is absolutely halting.

	Next we consider the states $\rho_{i,t}$ for different $i$.
	By definition of a quantum Turing machine, the transformations $M$ can apply to
	the entire state of the machine is given by some quantum channel.
	By Lemma~\ref{lem:non-increasing trace} we know that the trace distance between the
	entire state of the machine for separate instances of the catalytic tape can only decrease
	by this quantum channel.
	Therefore we know that if two instances start out to be orthogonal and end to be orthogonal,
	they have to remain orthogonal through the entire calculation. 
\end{proof}

\begin{lemma}
\label{lem:average}
	Let $M$ be a quantum catalytic Turing machine with
	work space $s$ and catalytic space $c$,
	let $\{\rho_i\}_i$ form an orthonormal basis for
	$D(\mathcal H_c)$, and define $T_{max}(M,x,\rho)$ to be the maximum
	runtime of machine $M$ on input $x$ on starting catalytic tape $\rho$.
	Then
	$$\mathbb{E}_i[T_{max}(M,x,\rho_i)] \leq 2^{O(s)}$$
\end{lemma}

\begin{proof}
	Our catalytic machine is defined by a $\SPACE[O(s)]$ machine,
	defined by a tape of length $O(s)$ and an internal machine of size $O(1)$,
	which acts on $\mathcal{H}_s$ and $\mathcal{H}_c$, which
	can be addressed into using $\log s$ and $\log c$ bits respectively.
	Since these quantities plus the Hilbert spaces $\mathcal{H}_s$ and $\mathcal{H}_c$
	define the dimensionality of our machine, by Lemma~\ref{lem:orthogonality-average} we have that
	$$\sum_{\rho \in \{\rho_i\}} T_{max}(M,x,\rho) \leq \mathcal O(2^{ 2(s + c + O(s) + O(1) + \log s + \log c)})$$
	  and therefore the lemma follows because $|\{\rho_i\}| \leq 2^{2c}$ and
	$2 (s + O(s) + O(1) + \log s + \log c )= O(s)$.
\end{proof}

\noindent
This already gives us a nice containment for our $\QCSPACE[s,c]$ classes.

\begin{corollary}
\label{thm:QCL_in_ZQP}
  $\QCLM \subseteq \mathsf{ZQP}$
\end{corollary}

\begin{corollary}
\label{thm:BQCL_in_BQP}
  $\BQCLM \subseteq \mathsf{BQP}$
\end{corollary}

\subsection{Equal running times}

We now take a further leap, showing that the initial catalytic tape does not
affect the (distribution of the) runtime of our machine $M$ for a fixed input $x$.

We can first show that given $M$ and only one single copy of a state
$\eta \in \mathcal H_c$, this probability distribution can be approximated up
to arbitrary precision for any $x$.

\begin{lemma}
\label{lem:approx_T}
	Given catalytic Turing machine $M$ and a single copy of a quantum state
	$\eta \in \mathcal H_c$, $T(M, x, \eta)$ can be approximated up to arbitrary
	precision for any $x$.
\end{lemma}

\begin{proof}
	Because $M$ is a quantum catalytic Turing machine it has to reset the quantum state initialized in its catalytic tape perfectly. Therefore we can use the following approach: first fix some input $x$, then run the catalytic machine given $x$ as input and $\eta$ on its catalytic tape and record the running time. When the machine halts, $\eta$ should be returned in the catalytic tape. This means the test can be performed again given the same inputs. This test can be run arbitrarily often giving an arbitrary approximation to $T(M,x,\eta)$. 
\end{proof}

This gives us the following observation about states with different halting times:

\begin{lemma}
\label{lem:orthogonality}
	Let $M$ be a quantum catalytic Turing machine, and let $\rho_1,\rho_2 \in D(\mathcal H_c)$.
	Assume there exists $x\in \{0,1\}^n$ such that $T(M, x, \rho_1)\neq T(M, x, \rho_2)$.
	Then $||\rho_1 - \rho_2||_1 = 1$, where $||\cdot||_1$ is the trace distance.
\end{lemma}
\begin{proof}
	The Helstrom bound states that the optimal success probability of any state discrimination protocol given one copy of an unnown state is:
	\[
		P_{success} = \frac{1}{2} + \frac{1}{2} \cdot ||\rho_1 - \rho_2||_1 
	\]
	By Lemma~\ref{lem:approx_T}, we know that $T(M, x, \rho)$ can be approximated to any precision with only one copy of $\rho$. Given a copy of either $\rho_1$ or $\rho_2$ at random, one can estimate $T(M,x,\rho)$ and perfectly discriminate between the cases $\rho=\rho_1$ and $\rho=\rho_2$ giving a protocol with $P_{success} = 1$. Therefore it follows that
	$$\frac{1}{2} + \frac{1}{2}||\rho_1 - \rho_2||_1 = 1$$
	and hence $||\rho_1 - \rho_2||_1 = 1$.
\end{proof}
Lemma~\ref{lem:orthogonality} is sufficient to show that the halting time of a
quantum catalytic Turing machine is independent of the initial state in
the catalytic tape:

\begin{theorem}
\label{thm:catalytic_independent_time}
	Let $M$ be a quantum catalytic Turing machine with $s$-qubit work space and $c$-qubit catalytic space, and let $x \in \{0,1\}^n$.
	Then there exists some value $t := t(n)$ such that $T(M, x, \rho) = t$
	for all $\rho \in D(\mathcal H_c)$.
\end{theorem}
\begin{proof}
	Assume for contradiction that there exist $\rho_1, \rho_2$ such that
	$T(M, x, \rho_1)\neq T(M, x, \rho_2)$. By Lemma~\ref{lem:orthogonality}
	it holds that $||\rho_1 - \rho_2||_1 = 1$.
	Consider the state $\rho' = \frac{1}{2} \rho_1 + \frac{1}{2} \rho_2$, and
	note that only one of $T(M, x, \rho') = T(M, x, \rho_1)$ or
	$T(M, x, \rho') = T(M, x, \rho_2)$ can hold, by transitivity.
	Without loss of generality, let us assume $T(M, x, \rho')= T(M, x, \rho_2)$, thereby $T(M,x,\rho') \neq T(M,x,\rho_1)$
	and so $||\rho' - \rho_1||_1 = 1$ by Lemma~\ref{lem:orthogonality}.
	However, by definition we have that
	$$||\rho' - \rho_1||_1 = ||(\frac{1}{2} \rho_1 +\frac{1}{2} \rho_2) - \rho_1||_1 = \frac{1}{2}$$
	which is a contradiction.
\end{proof}

Putting \cref{lem:average} and \cref{thm:catalytic_independent_time} together
immediately shows that the runtime of $M$ is bounded by a polynomial in $n$
for every input $x$ and initial catalytic state $\rho$:
\begin{theorem}
\label{thm:QCL_polytime}
	Let $M$ be a quantum catalytic Turing machine with work space $s$ and
	catalytic space $c$.
	Then the maximum halting time is bounded by $2^{\mathcal O(s)}$.
\end{theorem}

This strengthens Corollary~\ref{thm:QCL_in_ZQP} to
remove the randomness in the output probability; this is the
quantum equivalent of showing $\CL \in \P$, considered the
holy grail of open problems in classical catalytic computing:

\begin{corollary}
\label{thm:QCLM_in_EQP}
  $\QCLM \subseteq \mathsf{EQP}$
\end{corollary}

\subsection{Turing machines and circuits}

We finally prove Theorem~\ref{thm:cqtm_cqckt} and show the
equivalence of our two definitions of quantum catalytic machines.
To do this, we observe, without proof, that
Theorem~\ref{thm:catalytic_independent_time} extends to any
\textit{classical observable feature} of the initial catalytic state by the same proof.
We will apply this to one other aspect, namely the transition applied at
a given timestep $t$:

\begin{lemma}
\label{thm:catalytic_independent_transitions}
	Let $M$ be a quantum catalytic Turing machine, and let $x \in \{0,1\}^n$. Then
	for every time $t$, there exists a fixed operation $g$ applied by $M$
	at time $t$ for every $\rho \in \mathcal H_c$.
\end{lemma}
\noindent
This is sufficient to prove Theorem~\ref{thm:cqtm_cqckt}:

\begin{proof}[Proof of Theorem~\ref{thm:cqtm_cqckt}]
    We only prove the equivalence between $\QCSPACEC$ and $\QCSPACEM$;
    the same proof applies to $\BQCSPACEC$ and $\BQCSPACEM$.
    Certainly $\QCSPACEC[s,c]$ is contained in $\QCSPACEM[O(s),O(c)]$, since
    $\QCSPACEC$ circuits are $\SPACE[O(s)]$ uniform and
    can be directly simulated by a $\QCSPACEM$ machine. 

    Conversely, given a $\QCSPACEM[s,c]$ machine $M$, we wish to find an equivalent quantum catalytic circuit in $\QCSPACEC[O(s),O(c)]$. For this, we transform the transition function of the quantum Turing machine into a quantum channel; since the transition only takes a finite number of (qu)bits as input, this can be always be done, and we have our transitions act on the same space $\mathcal{H}_s \otimes \mathcal{H}_c$ as $M$. Then, by using a method similar to that from the proof of Lemma \ref{lemma:cl_obliv_no_halt}, to make the machine oblivious, the tape head movement of the quantum Turing machine will be fixed. If our circuit is the transition function channel copied to all locations where the tape heads end up, we completely simulate the quantum Turing machine. We know that $T_{max}(M, x, \rho)$ is always at most $2^{O(s)}$ for a machine $M$ by \Cref{thm:QCL_polytime}, and so the number of such transition function channels is also at most $2^{O(s)}$. Therefore, we can simulate $M$ using a quantum circuit of length $2^{O(s)}$ as claimed.
\end{proof}

As an afterword, we also resolve one other aspect of our initial definition of
quantum catalytic space, namely the requirement that the output state be the
same for every initial catalytic state. As mentioned above, Lemma~\ref{lem:approx_T}
extends to all classically observable characteristics, but a similar argument
clearly holds for approximating the output state as well:

\begin{lemma}
\label{lem:approx_out}
	Given catalytic Turing machine $M$ and a single copy of a quantum state
	$\eta \in \mathcal H_c$, the output qubit $\ket{\phi_{out}}$ can be approximated
    up to arbitrary precision for any $x$.
\end{lemma}

Thus again we can appeal to the instistinguishability of nearby catalytic states
to claim that $\ket{\phi_{out}}$ must be equal for all inital $\ket{\tau}$.

%% file: sections/TC1_in_QCL_new_old_final_final.tex
\section{Simulation of $\TCone$}

\newcommand{\pluseq}{\mathrel{+}=}

In this section we show that $\QCL$ can simulate Boolean threshold circuits.
As in the classical world, the ability to simulate $\TCone$ is also a reason
to believe that catalytic logspace is strictly more powerful than logspace.
This follows from the fact that $\QL = \PL$~\cite{watrous1998space}, which
is itself contained in $\TCone$:

\begin{lemma}
    $\QL \subseteq \TCone$
\end{lemma}

Since $\TCone$ can compute powerful functions such as determinant,
this containment is largely believed to be strict. Thus \Cref{thm:tcone}
gives us a candidate class of problems for separating $\QL$ from $\QCL$.

\subsection{Reversibility and obliviousness}

In \cite{buhrman2014computing} the authors
showed that $\TCone$ can be simulated by \textit{transparent register programs}, which
themselves are computable in $\CL$; thus our goal is to extend the $\CL$ simulation of
transparent programs to $\QCL$. More broadly, we show that \textit{reversible, oblivious,
time-bounded} $\CL$ is enough to simulate transparent programs, and such a model
is structured enough that, while we cannot show that all of $\CL$ is in $\QCL$,
we can at least prove the containment for this small fragment.

We first make the following definitions which we use for our simulations. 
We begin by recalling a result of Dulek~\cite{Dulek15} which shows
that catalytic Turing machines can be made \textit{reversible}
(see c.f.~\cite{CookLiMertzPyne25} for a proof)
\begin{theorem}
    \label{theorem:reversible_CTM}
    For every catalytic machine $M$ with space $s$ and catalytic space $c$, there exist catalytic machines $M_{\rightarrow}$, $M_{\leftarrow}$ with space $s + 1$ and catalytic space $c$ such that for any pair of configurations $(\tau_1,v_1)$, $(\tau_2,v_2)$ of $M_{\rightarrow}$ and $M_{\leftarrow}$, if $M_{\rightarrow}$ transitions from $(\tau_1,v_1)$ to $(\tau_2,v_2)$ on input $x$, then $M_{\leftarrow}$ transitions from $(\tau_2,v_2)$ to $(\tau_1,v_1)$ on input $x$.
\end{theorem}

We will also need to consider \textit{oblivious} machines, i.e. ones
where the tape head movement is solely a function of the input length $|x|$
and does not depend at all on the content of the catalytic tape c.
While any Turing machine can be made oblivious, it requires relaxing the
definition of obliviousness to not forcing the machine to halt at the
same time on every input; we simply require that every machine that continues
to run carries out its execution in an oblivious manner. We will bar
this restriction in this section.

\begin{definition}
    We say that a $\CL$ machine is \textit{totally oblivious} if
    the following holds. Let $t,q,h$ be special registers on the free work
    tape, all initialized to 0, representing the time, state,
    and tape heads of the machine.
    At each point in time our machine consist of one \textit{mega-step}:
    for every setting of $t,q,h$ there is a fixed transformation,
    computable in logspace,
    which the machine applies to the catalytic tape and to $q,h$, and a mega-step consists
    of applying each of these operations, conditioned on the values of $t,q,h$
    on the free work tape, in order. At the end of every mega-step we increment $t$,
    and our machine halts iff $t$ reaches a predetermined step $T$.
\end{definition}

Totally oblivious machines are ones that in essence apply the same bundle
of transformations at every time step, with the information about
which one to to actually apply being written on the free work tape,
and the halting behavior being determined only by the clock.

Such machines are clearly in poly-time bounded $\CL$
(see c.f.~\cite{CookLiMertzPyne25} for a discussion of this class),
since the clock must fit on the free work tape. This causes issues when
we seek total obliviousness in tandem with reversibility;
in general it is not known, and is highly unlikely, that a polynomially
time-bounded Turing machine can be made reversible while remaining
polynomially time-bounded.

However, there is an important class of algorithms which is
both reversible and totally oblivious: \textit{clean register programs}.
For our purposes we will use a very restricted version of clean register
programs (see c.f. \cite{Mertz23} for a discussion).

\begin{definition}
    A \textit{register program} $\mathcal{P}$ is a list of instructions
    $P_1 \ldots P_t$ where each $P_i$ either has the form $R_j \pluseq x_k$
    for some input variable $x_k$ or has the form $R_j \pluseq q_i(R_1 \ldots R_m)$
    for some polynomial $q_i$. A register program
    \textit{cleanly computes} a value $v$ if for any
    initial values $\tau_1 \ldots \tau_m$, the net result
    of running $\mathcal{P}$ on the registers $R_1 \ldots R_m$,
    where each $R_j$ is initialized to the value $\tau_j$,
    is that $R_1 = \tau_1 + v$ and $R_j = \tau_j$ for all $j \neq 1$.
\end{definition}

If we think of these registers as being written on the catalytic tape,
it is clear that clean register programs are totally oblivious,
as the instruction at every moment in time is based only on the
timestep. This is nearly immediate, although we note a few minor complications here.
We need to preprocess the catalytic tape to ensure our registers
have values over the same ring as our register program; for example,
if we represent numbers mod $p$ using $\lceil \log p \rceil$ bits,
some initial values will exceed $p$. This can be handled obliviously by observing that
for either $\tau$ or $\overline{\tau}$, half the registers are already correct,
and so we take one full pass over $\tau$ to keep a count of which case we are in,
store this as a bit $b$ (1 iff we need to flip $\tau$),
and XOR $\tau$ with $b$ at the beginning and end of the computation.
We subsequently ignore all blocks which are initialized
to improper values; when we go to operate on register $R_j$, say, as we
obliviously pass over the whole catalytic tape we will count how many \textit{valid}
registers we have seen, and act only when we see the counter reach $j$.

Besides being totally oblivious, however, such programs are also \textit{reversible},
as every step of the form $R_j \pluseq c$ can be inverted by a step of the form
$R_j -= c$. Thus such programs appear highly constrained in terms of what they can
and cannot achieve. Nevertheless, such programs are sufficient to compute $\TCone$.

\begin{lemma}[\cite{buhrman2014computing}] \label{lem:register-program}
    Let $L$ be a language in $\TCone$. Then $L$ can be decided by
    a clean register program, and, hence, by a totally oblivious
    reversible $\CL$ machine.
\end{lemma}

\subsection{Simulation by $\QCL$ machines}
We now show that reversibility plus total obliviousness is
sufficient for simulation by $\QCL$.

\begin{lemma}
\label{lem:oblivousCL_in_QCL}
    Let $L$ be a language which can be computed be a totally
    oblivious reversible $\CL$ machine.
    Then $L \subseteq \QCL$.
\end{lemma}

\begin{proof}
    Let $M$ be a totally oblivious reversible $\CL$ machine.
    We will treat our quantum catalytic tape as a superposition
    over classical catalytic tapes, i.e. a superposition over
    computational basis states. It is thus sufficient to show
    that the operation of machine $M$ can be simulated by a
    fixed quantum circuit containing Toffoli gates, as such
    a circuit will correctly operate on each of our catalytic
    basis states in each branch of the superposition.

    By total obliviousness, every step that $M$ takes is
    a fixed transformation conditioned on the value of $t$, $q$,
    and $h$; since we additionally know that such a step is
    reversible, it must be isomorphic to a Toffoli gate applied
    to a fixed position of the catalytic tape conditioned on
    some fixed mask applied to $t$, $q$, and $h$, and furthermore
    each transformation can be computed by our logspace
    controlling machine.
    Since these operations are fixed for each
    timestep, we can move $t$ to our space controlling machine
    and have it construct a circuit, comprised of Toffoli gates
    on $q$, $h$, and the catalytic tape, of polynomial length.
\end{proof}

This is sufficient to prove our main result for this section:
\begin{proof}[Proof of \Cref{thm:tcone}]
    Combine \Cref{lem:register-program} with \Cref{lem:oblivousCL_in_QCL}.
\end{proof}

%% file: sections/CL_in_DQC1.tex
\section{Simulating catalytic space in $\DQCOne$}
Lastly we will discuss the relationship between catalytic computing
and a pre-existing yet closely related quantum model, namely the
one clean qubit setting.
We will introduce the model and then prove that it can simulate unitary $\QCL$.
Finally we will show that classical $\CL$ is also contained in the one clean qubit model.

\subsection{One clean qubit model}
In the one-clean qubit model, first introduced by Knill and Laflamme~\cite{Knill_1998}, a quantum machine is given a single input qubit initialized in the zero state and $n$ qubits initialized in the maximally mixed state.
We will formalize the definition of this computational model:

\begin{definition}[One clean qubit]
Let $\{Q_x\}_x$ be a log-space uniform family of unitary quantum circuits.
The \emph{one clean qubit model} is a model of computation in which $Q_x$
is applied to the $n+1$-qubit input state
\[
    \rho = \ketbra{0}\otimes \frac{I_n}{2^n},
\]
where $n = |x|$ and $I_n$ operator is the identity on $n$ qubits.
After execution of $Q_x$ the first qubit is measured, giving output probabilities:
\begin{align*}
    p_0 &= 2^{-n}\Tr[(\ketbra{0}\otimes I)Q_x(\ketbra{0}\otimes I)Q_x^\dagger ],\\
    p_1 &= 1 - p_0
\end{align*}
\end{definition}
\begin{remark}
Two points stand out in this definition.
First, note that $Q_x$ are unitary circuits, and hence do not allow
intermediate measurements; such measurements would allow for resetting the
qubits initialized in the maximally mixed state, making the model significantly stronger.
Second, in this paper we consider log-space uniform families of unitary circuits,
rather than the more common deterministic polynomial-time uniform families,
in order to align more closely with the $\QCL$ model that we study. 
\end{remark}

The one-clean qubit model is a probabilistic model of computation, and hence we typically
talk about computing a function $f(x)$ in terms of success probability
for computing $f(x)$ being bounded away from $1/2$. The exact bound on the error probability does
not matter; while we often use $2/3$ in defining e.g. $\BQP$, even a $1/\poly(n)$
gap is sufficient as there we can employ standard error-correction to boost
our success, namely by running the algorithm multiple times. However, this
is not known to be possible in the one-clean qubit model, as such a machine
can only reliably run once.

\begin{definition}[\cite{Knill_1998,shepherd2006computationunitariespurequbit}]
$\DQCOne$ is the set of all languages $L = (L_{yes}, L_{no})\subset \{0,1\}^* \times \{0,1\}^*$ for which there exists a one-clean qubit machine $M$ and a polynomial $q(n)$ that on input $x \in L$ of length $n = |x|$,
\begin{itemize}
    \item if $x \in L_{yes}$ then the output probability $p_1 \geq \frac{1}{2} + \frac{1}{q(n)}$
    \item if $x \in L_{no}$ then the output probability $p_0 \geq \frac{1}{2} + \frac{1}{q(n)}$
\end{itemize}
\end{definition}

On the other hand, somewhat surprisingly the one-clean qubit model is robust to
the number of clean qubits allowed, up to a logarithmic number:

\begin{lemma}[\cite{shor2008estimating}]\label{lem:dqck}
$\mathsf{DQC}_k = \mathsf{DQC}_1$ for $k = \mO(log(n))$,
where $\mathsf{DQC}_k$ means having access to $k$ clean qubits instead of one.
\end{lemma}

\subsection{Containment of unitary $\QCL$ in $\DQCOne$}
We now move on to establishing a formal connection between $\QCL$ and $\DQCOne$.
A $\mathsf{QCL}$ machine is allowed to apply intermediate measurements
to its quantum tape as well as its catalytic tape, which is not possible
in $\DQCOne$; however, if we restrict the $\mathsf{QCL}$ machine
to not make any intermediate measurements we can show that such a machine
can in fact be simulated by the one-clean qubit model.

\begin{definition}[$\mathsf{Q_{U}CL}$]
A $\mathsf{Q_{U}CL}$ machine is a $\QCL$ machine in which the quantum circuit is unitary.
In the final step of the unitary the $\mathsf{Q_{U}CL}$ machine measures the first qubit,
which then gives the outcome of the calculation.
Similarly we define $\mathsf{BQ_{U}CL}$ to be $\BQCL$ with the unitary restriction.
\end{definition}

Using this definition we can give the following proof of containment:
\begin{proof}[Proof of \Cref{thm:bqucl_in_dqc1}]
Let $C$ be a log-space uniform $\mathsf{BQ_{U}CL}$ quantum channel.
Since $C$ is unitary up until the last measurement step, it preserves all possible
density matrices from the catalytic tape, and in particular it preserves
the maximally mixed state $I_n$. Let $U$ be the unitary part of $C$.
The action of $U$ on the work-tape and the catalytic tape,
with the catalytic tape initialized in $I_n$, is:
\[
    U \ket{0}\bra{0}_w \otimes \frac{I_n}{2^n} U^\dagger = (\sqrt{p_0}\ket{0}\bra{0}_{w_0}\ket{\psi_0}\bra{\psi_0}_w + \sqrt{p_1} \ket{1}\bra{1}_{w_0}\ket{\psi_1}\bra{\psi_1}_w) \otimes \frac{I_n}{2^n}
\]
with $|p_1| \geq 2/3$ in a 'yes' instance and $|p_0|\geq 2/3$ in a 'no' instance.
Note that this calculation is of the exact form of a $\log(n)$-clean qubit machine
and that the output probabilities are a constant bounded away from $1/2$;
hence this problem is in $\mathsf{DQC_k}$, and by Lemma~\ref{lem:dqck}
is therefore in $\mathsf{DQC_1}$ 
\end{proof}

\subsection{Containment of $\CL$ in $\DQCOne$}

We aim to show that $\CL \subseteq \DQCOne$.
The idea is that $\CL$, as per \Cref{theorem:reversible_CTM}, can always be made reversible.
While as discussed before we cannot maintain reversibility and total obliviousness,
a $\CL$ machine can also always be made `almost oblivious' while maintaining reversibility;
the tape head movements are independent of the input, but the machine does not know when
to halt. Instead, after any given amount of time, we know that the machine has halted on
a fraction $1/\poly(n)$ of possible initial catalytic states.
Since the $\DQCOne$ model can be interpreted as sampling from a uniform distribution of computational basis states, this shows the probability of finding the correct output is $1/2 + 1/\poly(n)$, which is sufficient for the proof.

\begin{definition}
    A \emph{non-halting reversible oblivious} catalytic Turing machine is a reversible oblivious catalytic Turing machine that need not halt absolutely. In particular, for every input $x$ and initial catalytic state $c$ there exists a time $t(x, c)$ where the correct output has been written to the output tape and the catalytic tape has been reset to its initial state. In addition, the output state has an additional binary cell that indicates whether or not the output has been determined yet, or is still `unknown' by the machine.
\end{definition}

\begin{definition}
    We say a reversible oblivious catalytic Turing machine \emph{halts with polynomial success probability} if there exists polynomials $p, q$ such that for any valid input $x$ to a promise problem, after time $p(|x|)$ the output tape of the catalytic Turing machine contains the correct output to the problem on a fraction of at least $1 / q(|x|)$ when the initial catalytic tapes are taken uniformily at random. After time $p(|x|)$, the output tape of the catalytic Turing machine never contains the wrong answer, but it may leave the output undetermined.
\end{definition}

We show that any $\CL$ machine can be transformed into a reversible oblivious catalytic Turing machine that halts with polynomial success probability. 

\begin{lemma}
    \label{lemma:cl_obliv_no_halt}
    Any catalytic Turing machine $M$ that has a logarithmic clean space and polynomial size catalytic tape can be turned into a non-halting oblivious reversible catalytic Turing machine $M^o$ with a logarithmic clean tape and polynomial catalytic tape.
\end{lemma}
\begin{proof}
    By \cite{Dulek15,CookLiMertzPyne25}, $M$ can always be assumed to be reversible. We claim we can also make $M$ oblivious by sacrificing the condition that $M$ is absolutely halting. This also interferes with what is meant by the machine being catalytic, but the new machine no longer needs to be catalytic. 
    
    To make the machine oblivious, we make two modifications. The first applies to operations on the clean tape. The second applies to operations on the catalytic tape. On the clean tape, we double the size of the clean tape of $M$, breaking it up into pairs. The first entry of the pair stores the original data while the second keeps track of the position of where the tape head is `supposed' to be.
    Then by iterating over all positions on the clean tape of the Turing machine in every step of the original Turing machine, operations on the clean tape of the Turing machine can be made oblivious. Similarly, for operations on the catalytic tape, we can maintain an additional part of the clean tape that keeps track of the position of the catalytic tape head position. By iterating over all possible positions of the catalytic tape head and checking if the tape head is `really there', we can make catalytic tape operations oblivious.
\end{proof}

We call the machine formed this way $M^o$ for oblivious $M$. 
Since the catalytic and clean tape are no more than polynomial length, this procedure adds at most a polynomial factor to the runtime. However, since the runtime of $M$ may be super-polynomial and an oblivious machine has the same runtime for all inputs $x$ of the same length and catalytic tapes $c$, the machine does not have enough clean space to keep a clock to know whether or not it has terminated. This means we cannot assume it to be halting. However, we can show that it is halting with sufficient probability:

\begin{lemma}
    \label{lemma:cl_polynomial_halt}
    For any language $L$ in $\CL$ that is recognized by a catalytic Turing machine $M$, there exists a reversible oblivious catalytic Turing machine $N$ that halts with polynomial probability 
    that also recognizes $L$. Furthermore, $N$ also uses $O(\log |x|)$ clean space and polynomial catalytic space.
\end{lemma}
\begin{proof}
    We observe that $M^o$ in Lemma \ref{lemma:cl_obliv_no_halt} is simulated step-by-step, meaning that not only do we reach the same outcome, but up to a fixed transformation and a slower runtime, $M^o$ passes through the same intermediate states. If we consider a modified version of $M$ in the first place, we can ensure that the machine halts with polynomial success probability. We modify the original machine $M$ to form the machine $M'$ in the following ways:
    \begin{enumerate}
        \item The machine $M'$ repeats the original machine $M$ $l(n) = 2^{s(n)}$ times where $s(n)$ is the length of the clean tape.
        Each repetition is called a cycle.
        \item The space used for writing the output of $M$ originally is extended by one additional bit. This bit starts out in $0$, signifying the output is `undetermined'. These to bits together are called the output state.
        \item After each cycle, the machine cleans itself. This means that it resets the clean space to the all $0$ state, reversing the computation except the output state, which is left unaltered except for the first cycle.
        \item After the first cycle, the correct output of the computation is written to the output tape and the second bit of the output state is flipped, signifying the output is `determined'. For every subsequent cycle, a counter that counts up to $2l(n)$ is incremented by 1, but the output tape is left unchanged.
    \end{enumerate}

    We claim that the machine $M'^o$ halts with polynomial probability. Suppose that the runtime of a cycle of $M'$ on input $x$ and catalytic tape $c$ is $f(x, c)$. The expected runtime of a cycle of $M'$, $f(x)$, over a uniform distribution of catalytic tapes for fixed $x$ is at most $l(n)$ by close analysis of the polynomial expected time bound given in \cite{buhrman2014computing}. 
    Let us examine the output at time $t(x) = t(|x|) = l(n) + 1$.
    Then, by Markov's inequality on a uniform distribution of possible catalytic tapes
    \begin{align*}
        \mathbb{P}(\text{output undetermined}) 
        &\leq \frac{f(x)}{t(x)} \\
        &\leq \frac{l(n)}{t(x)} \\
        &= \frac{l(n)}{l(n) + 1} = 1 - \frac{1}{l(n) +1}
    \end{align*}
    This means that at time $t(n)$, the probability that an output is written to the clean tape is at least $1 / (l(n) + 1) = 1 / \poly(n)$. After the first cycle, the output must be correct. However, afterwards, we have no control over what is written onto the output tape. In making the machine reversible and oblivious, it may later change the value in the output tape, including incorrect values. This is why we repeat each cycle many times. This is $M'$ stalling to preserve the correctness of the output. Since the cycle is repeated $2l(n)$ times and each cycle uses time at least 1 to increment the counter, this means that at time $t(x) = l(n) + 1$ the value in the output tape is guaranteed to be correct or undetermined.

    Let $N = M'^o$. Then $N$ is reversible, oblivious and halts with polynomial probability. Since $t(x) = l(n) + 1$ and $l(n)$ or any upper bound on $l(n)$ (which is sufficient) is readily computable, this completes the proof.
\end{proof}

\noindent
This completes all technical components necessary to show that $\CL \subseteq \DQCOne$.

\begin{proof}[Proof of \Cref{thm:cl_in_dqc1}]
    The maximally mixed state of $\DQCOne$ can be interpreted as uniformly randomly sampling computational basis states. If we take these basis states to be the catalytic tape and use the fact that $\DQCOne$ is unchanged if we allow a logarithmic number of clean qubits, then we can run the machine $N$ from Lemma \ref{lemma:cl_polynomial_halt} by using unitary gates instead of reversible, oblivious operations. When we measure the output bit at the end, we get either an indeterminate state or the correct output with certainty. If we get an indeterminate state, we output a random bit and thus output the correct answer with probability $1/2$. If not, then we output the correct answer, which occurs with probability at least $1/\poly(n)$.
\end{proof}

%% file: sections/acknowledgements.tex
\section*{Acknowledgements}

We thank \cite{mathoverflow_bound} for pointing us to \cite{alon_bound}. SS acknowledges support from the Royal Society University Research Fellowship.